\documentclass[11pt]{article}
\usepackage[margin=1in]{geometry}
\usepackage[colorlinks=true,
            linkcolor=blue,
            citecolor=black,
            urlcolor=blue]{hyperref}

\usepackage{setspace}

\usepackage{subcaption}

\usepackage{fixmath}
\usepackage{bm}
\usepackage{amsbsy}
\usepackage{color}
\usepackage{verbatim}
\usepackage{multirow}
\usepackage{amssymb}
\usepackage{amsthm}
\usepackage{amsmath}
\usepackage{array}
\usepackage{mathtools}

\allowdisplaybreaks

\usepackage{graphicx}
\usepackage{tikz}
\usetikzlibrary{positioning}
\usepackage{xcolor}

\usepackage{thm-restate}
\usepackage{algpseudocode}
\usepackage{algorithm}
\usepackage{algorithmicx}

\newtheorem{theorem}{Theorem}
\newtheorem{lemma}{Lemma}

\newtheorem{definition}{Definition}

\newtheorem{remark}{Remark}
\newtheorem{claim}{Claim}

\newcommand{\size}[1]{\ensuremath{|#1|}}
\newcommand{\ceil}[1]{\ensuremath{\lceil#1\rceil}}

\newcommand{\Ceil}[1]{\ensuremath{\left\lceil#1\right\rceil}}

\newcommand{\lrA}[1]{\ensuremath{\left(#1\right)}}
\newcommand{\lrB}[1]{\ensuremath{\left[#1\right]}}
\newcommand{\lrC}[1]{\ensuremath{\left\{#1\right\}}}

\usepackage{natbib}

\def\OPT{\mbox{OPT}}
\def\MST{\mbox{MST}}

\def\T{\mathcal{T}}

\def\OPT{\mbox{OPT}}
\def\MST{\mbox{MST}}

\def\capacity{k}

\title{An Improved Approximation Algorithm for the Capacitated Arc Routing Problem}

\author
{
Jingyang Zhao\\
University of Electronic Science and Technology of China\\
\texttt{jingyangzhao1020@gmail.com}
\and
Mingyu Xiao\footnote{Corresponding author}\\
University of Electronic Science and Technology of China\\
\texttt{myxiao@uestc.edu.cn}
}

\date{}

\begin{document}

\maketitle

\begin{abstract}
The Capacitated Arc Routing Problem (CARP), introduced by Golden and Wong in 1981, is an important arc routing problem in Operations Research, which generalizes the famous Capacitated Vehicle Routing Problem (CVRP).
When every customer has a unit demand, the best known approximation ratio for CARP, given by Jansen in 1993, remains $\frac{5}{2}-\frac{1.5}{k}$, where $k$ denotes the vehicle capacity.
Based on recent progress in approximating CVRP, we improve this result by proposing a $(\frac{5}{2}-\Theta(\frac{1}{\sqrt{k}}))$-approximation algorithm, which to the best of our knowledge constitutes the first improvement over Jansen's bound.
 
\end{abstract}

\maketitle

\section{Introduction}
The \textsc{Capacitated Arc Routing Problem} (CARP), introduced by ~\citet{GoldenW81}, is one of the most famous arc routing problems~\citep{van2015chapter}. In this problem, we are given an undirected graph $G = (V\cup\{v_0\}, E)$, where each edge $e \in E$ has an associated nonnegative cost $c(e) \in \mathbb{R}_{\geq 0}$ and a nonnegative integer demand $d(e) \in \mathbb{Z}_{\geq 0}$.
Edges with positive demand are referred to as \emph{customers}.
A fleet of identical vehicles, each with capacity $\capacity \in \mathbb{Z}_{\geq 1}$, is initially located at a designated depot vertex $v_0 \in V$.
The objective is to compute a set of routes (walks), each starting and ending at the depot, such that all customers are served, the total demand served by each vehicle does not exceed its capacity, and the total cost of the routes is minimized.
It is typically assumed that each customer is served by exactly one vehicle~\citep{jansen1993bounds,wohlk2008decade}.
When each edge's demand is either 0 or 1, the problem is referred to as the \emph{equal-demand CARP}; otherwise, it is known as the \emph{general CARP}.

CARP generalizes many famous routing problems, e.g., the \textsc{Chinese Postman Problem} (CPP)~\citep{eiselt1995arc1}, the \textsc{Rural Postman Problem} (RPP)~\citep{eiselt1995arc}, the \textsc{Traveling Salesman Problem} (TSP)~\citep{christofides1976worst}, and so on. Moreover, if the demands are defined for the vertices instead of the edges, CARP reduces to the \textsc{Capacitated Vehicle Routing Problem} (CVRP)~\citep{dantzig1959truck}, which is a representative problem in the area of node routing. Since CARP has various application in road networks, e.g., snow plowing~\citep{perrier2007survey}, waste collection~\citep{fernandez2016collaboration}, newspaper delivery~\citep{corberan2015arc}, and so on, it has been extensively studied in the areas of Operations Research and Computer Science~\citep{corberan2010recent,mourao2017updated}. A recent survey could be found in~\citep{corberan2021arc}.

We focus on approximation algorithms for CARP. 
For any minimization problem, an algorithm is called a $\rho$-approximation algorithm if it computes a solution with an objective value not exceeding $\rho$ times the optimal value in polynomial time, where $\rho\geq 1$ is called the \emph{approximation ratio}.

\subsection{Related work}
It is well known that CARP can be solved in polynomial time when $k=1$ or $k=2$. 
However, since CARP generalizes CVRP, it is APX-hard for any $\capacity \geq 3$~\citep{AsanoKTT97}. Moreover, as general CARP also generalizes the \textsc{Bin Packing Problem}~\citep{jansen1993bounds}, it is NP-hard to approximate within a factor better than $\frac{3}{2}$.
When $\capacity = \infty$, CARP reduces to the RPP, which in turn generalizes the TSP. Hence, both RPP and TSP remain APX-hard~\citep{karpinski2015new}.
However, if the subgraph induced by the edges with positive demand is connected, then RPP reduces to the CPP, which is solvable in polynomial time~\citep{eiselt1995arc1}.

CARP and its related problems has been extensively studied in both theory~\citep{jansen1993bounds,wohlk2008approximation,van2014constant,van2017parameterized,van2020approximate} and practice~\citep{brandao2008deterministic,santos2009improved,mourao2009heuristic,wohlk2018fast}.
For heuristic algorithms, a wide range of methods have been developed, including combinatorial methods~\citep{wohlk2008approximation}, memetic search~\citep{tang2009memetic}, tabu search~\citep{brandao2008deterministic}, metaheuristic~\citep{chen2016hybrid}, simulated annealing~\citep{babaee2016solving}, integer
programming~\citep{belenguer2003cutting}, and so on. More results can be found be in~\citep{corberan2010recent,mourao2017updated,corberan2021arc}.
For approximation algorithms, we first review the results on CVRP. 

For the equal-demand CVRP, where each customer has unit demand, \citet{HaimovichK85} proposed an $(\alpha + 1 - \frac{\alpha}{k})$-approximation algorithm under the assumption that the number of customers is divisible by $k$. Here, $\alpha$ denotes the approximation ratio of the metric TSP, for which it is known that $\alpha = \frac{3}{2}$~\citep{christofides1976worst} and was later slightly improved to $\alpha = \frac{3}{2} - 10^{-36}$~\citep{KarlinKG21, DBLP:conf/ipco/KarlinKG23}.
Subsequently, \citet{altinkemer1990heuristics} removed the divisibility assumption while maintaining the same approximation ratio of $\alpha + 1 - \frac{\alpha}{k}$. For the general CVRP, \citet{altinkemer1987heuristics} proposed an $(\alpha + 2 - \frac{2\alpha}{k})$-approximation algorithm when $k$ is even.
Note that equal-demand (resp., general) CVRP is also referred to as the \emph{unit-demand} (resp., \emph{unsplittable}) CVRP. Another common variant is the \emph{splittable} CVRP, where a customer may be severed by multi-vehicles. 
However, as noted in the survey by \citet{corberan2021arc}, splittable service is less common in arc routing problems.

\citet{BompadreDO06} improved these results by an additive term of $\Omega(\frac{1}{k^3})$. More recently, \citet{blauth2022improving} achieved an improved approximation ratio of $\alpha + 1 - \varepsilon$ for the equal-demand case and $\alpha + 2 - 2\varepsilon$ for the general case, where $\varepsilon > \frac{1}{3000}$ when $\alpha = \frac{3}{2}$.
As a further improvement for the general case, \citet{uncvrp} achieved an approximation ratio of $\alpha + 1 + \ln 2 - \varepsilon$ for some positive constant $\varepsilon$.
In addition, \citet{DBLP:conf/mfcs/zhao} studied the case where the vehicle capacity $k$ is a fixed integer, and proposed a $(\frac{5}{2} - \Theta(\frac{1}{\sqrt{k}}))$-approximation algorithm for the equal-demand case, and a $(\frac{5}{2} + \ln 2 - \Theta(\frac{1}{\sqrt{k}}))$-approximation algorithm for the general case.

Approximation algorithms on multi-depot CVRP could be found in~\citep{tight,HarksKM13,zhao2025multidepot}.

For the equal-demand CARP, \citet{jansen1993bounds} proposed a $(\beta + 1 - \frac{\beta}{\capacity})$-approximation algorithm on metric graphs, where $\beta = \frac{3}{2}$ is the approximation ratio of the RPP~\citep{eiselt1995arc}.
For the general CARP, \citet{jansen1993bounds} gave a $(\beta+\frac{\capacity}{\capacity_s} - \frac{\beta}{\capacity_s})$-approximation algorithm on metric graphs, where $\capacity_s = \lceil \frac{\capacity}{2} \rceil$.
This yields an approximation ratio of $\beta + 2 - \frac{2\beta}{\capacity}$ when $\capacity$ is even, and $\beta + 2 - \frac{2\beta}{\capacity + 1}$ when $\capacity$ is odd.
\citet{wohlk2008decade} also obtained the same approximation ratios on metric graphs using a dynamic programming approach.
Notably, there is no approximation gap between metric and non-metric graphs for CARP, since \citet{van2014constant} showed that any $\rho$-approximation algorithm for CARP on metric graphs also yields a $\rho$-approximation algorithm on non-metric graphs.

It is worth noting that when $k = O(1)$, CVRP can be reduced to the \textsc{Minimum Weight $k$-Set Cover Problem}~\citep{chvatal1979greedy,hassin2005better,gupta2023local} in polynomial time $n^{O(k)}$, while preserving the approximation ratio (see~\citep{DBLP:conf/mfcs/zhao} for details).
This reduction also applies to CARP.
Therefore, based on the results of \citet{gupta2023local}, both CVRP and CARP admit an approximation ratio of $\min\{H_k - \frac{1}{8k}, H_k - \sum_{i=1}^{k} \frac{\log i}{8k i} \}$, where $H_k = \sum_{i=1}^k \frac{1}{i}$ denotes the $k$-th harmonic number.

Approximation algorithms on multi-depot CARP could be found in~\citep{yu2023exact}.

We observe that, in the early development of approximation algorithms, there was little to no gap between the approximation ratios for CVRP and CARP. However, as approximation algorithms for CVRP have advanced, the gap has grown significantly.
In particular, the metric TSP now admits an approximation ratio of $\alpha = \frac{3}{2} - 10^{-36}$~\citep{KarlinKG21, DBLP:conf/ipco/KarlinKG23}, while the best known approximation ratio for the RPP remains $\beta = \frac{3}{2}$~\citep{eiselt1995arc}.
Moreover, the best known approximation ratios for CARP are still those given by \citet{jansen1993bounds}.
This raises a natural question: can the techniques developed for CVRP be extended to CARP to improve upon the existing results?

\subsection{Our results}
In this paper, we focus on approximation algorithms for equal-demand CARP.
According to the result of \citet{van2014constant}, it suffices to consider metric graphs.
By extending the methods of analyzing the lower bounds of an optimal solution in~\citep{DBLP:conf/mfcs/zhao}, we propose a $(\frac{5}{2} - \Theta(\frac{1}{\sqrt{k}}))$-approximation algorithm, which constitutes the first improvement over the previous ratio of $\frac{5}{2} - \frac{1.5}{k}$ by \citet{jansen1993bounds}.
A summary of previous and our results for equal-demand CARP can be found in Table~\ref{res}.

\begin{table}[ht]
\small
\centering
\resizebox{0.7\textheight}{!}{
\begin{tabular}{cccccccc}
\hline
  $k$ & 3 & 4 & 5 & 6 & 7 & 8 & $\cdots$\\
\hline
  \multirow{2}{*}{Previous Ratio} & $\boldsymbol{1.792}$ & 2.051 & 2.200 & 2.250 & 2.286 & 2.313 & $\cdots$\\
  &\citep{gupta2023local} & \citep{jansen1993bounds} & \citep{jansen1993bounds} & \citep{jansen1993bounds} & \citep{jansen1993bounds} & \citep{jansen1993bounds}\\
\hline
  Our Ratio & ${1.889}$ & $\boldsymbol{2.000}$ & $\boldsymbol{2.086}$ & $\boldsymbol{2.143}$ & $\boldsymbol{2.184}$ & $\boldsymbol{2.215}$ & $\cdots$\\
\hline
\end{tabular}}
\caption{Previous and our approximation ratios for equal-demand CARP.}
\label{res}
\end{table}

Similar to the previous algorithm, our algorithm first computes an RPP tour, which is a solution to RPP, and then obtain a solution to CARP by partitioning the RPP tour. 
The key difference is that we use two different RPP tours.
The first is computed using the $\frac{3}{2}$-approximation algorithm for RPP~\citep{eiselt1995arc}. The second is obtained in a different way: we begin by finding a minimum-cost even-degree multi-graph that includes all edges with positive demand, using a minimum-cost perfect matching algorithm~\citep{schrijver2003combinatorial}; then, we connect the resulting components by greedily adding minimum-cost edges between them, similar to Kruskal’s algorithm~\citep{schrijver2003combinatorial}; last, we form an Eulerian graph by doubling the added edges, and then obtain an RPP tour by shortcutting the Eulerian graph.

\section{Preliminary}
We consider CARP on metric graphs. The input graph, denoted by $G=(V\cup\{v_0\},E)$, is an undirected complete graph, where $v_0$ is the depot and the edge cost function $c:E\to\mathbb{R}_{\geq0}$ satisfies $c(a,a)=0$, $c(a,b)=c(b,a)$, and the triangle inequality $c(a,h)\leq c(a,b)+c(b,h)$ for all $a,b,h\in V$.
Let $E^*\subseteq E$ denote the set of edges (customers) with positive demand. 
By making copies of vertices that are shared among multiple edges in $E^*$, we assume w.l.o.g.\ that the edges in $E^*$ are vertex-disjoint. 
Moreover, by the triangle inequality, we can move any vertex in $V$ that is not incident to an edge in $E^*$. Thus, we assume that each vertex in $V$ is incident to exactly one edge in $E^*$, which implies that $\size{E^*}=\frac{\size{V}}{2}$. We let $\size{V}=n$, and for any positive integer $t$, we define $[t]=\{1,2,...,t\}$. 
We assume that there is an unlimited number of vehicles at the depot, each with capacity $\capacity \in \mathbb{Z}_{\geq 1}$.

For any subset $E'\subseteq E^*$, define $\delta(E')=\sum_{(v_i,v_j)\in E'}\frac{c(v_0,v_i)+c(v_i,v_j)+c(v_0,v_j)}{2}$.
For any multi-graph subgraph $S$ of $G$, let $V(S)$ (resp., $E(S)$) denote the set of its vertices (resp., the (multi-)set of its edges). 
We define $c(S)=\sum_{e\in E(S)}c(e)$. Throughout the paper, we work with multi-edge sets, and the union of any two edge sets is taken with multiplicities.

A \emph{walk} $W$ in a graph, denoted by $(v_1,v_2,\dots, v_l)$, is a sequence of vertices where each consecutive pair $(v_i,v_{i+1})$ is connected by an edge, and vertices may appear multiple times.
We use $E(W)$ (resp., $V(W)$) to denote the multi-set of edges $\{(v_1,v_2),...,(v_{l-1},v_{l})\}$ (resp., the set of vertices $\{v_1,...,v_l\}$) and define the cost of the walk as $c(W)=\sum_{e\in E(W)}c(e)$.
A \emph{closed walk} is a walk in which the first and last vertices are the same, while a \emph{cycle} is a closed walk in which no other vertex is repeated.
Given a closed walk $W$, we can obtain a new walk $W'$ by skipping some vertices or edges along $W$; this process is called \emph{shortcutting}.
By the triangle inequality, shortcutting does not increase the cost, i.e., $c(W')\leq c(W)$.
A \emph{tour} or \emph{route} is a closed walk that starts and ends at the depot $v_0$, and an \emph{RPP tour} is a tour that traverses every edge in $E^*$ at least once.  
By the triangle inequality, we may assume that any RPP tour consists of exactly $\size{V}+1$ distinct edges, i.e., it forms a \emph{Hamiltonian cycle} in $G$. Let $H^*$ denote an optimal RPP tour in $G$, and then it is a minimum-cost Hamiltonian cycle in $G$ such that $E^*\subseteq E(H^*)$.

A solution to the CARP is a set of tours, where each tour corresponds to the route of a single vehicle and serves a set of edges (customers) with total demand at most $k$ in $E^*$.
By the triangle inequality, we may assume that each vehicle route forms a cycle and visits only the vertices of its served edges along with the depot. Then, each route $T$ can be represented as $(v_0, v_1, v_2, \dots, v_{2l-1}, v_{2l}, v_0)$, where it serves the set of customers $\{(v_{2i-1}, v_{2i}) \mid i \in [l]\}$, and $\sum_{i\in[l]} d(v_{2i-1}, v_{2i}) \leq k$. Hence, we define the set of customers served by route $T$ as $E^*_T= E^*\cap E(T)$ and also let $V_T=\{v_0,v_1,...,v_{2l}\}$.

CARP is formally defined as follows.

\begin{definition}[CARP]\label{def}
Given an undirected complete graph $G=(V\cup\{v_0\},E)$, a set of vertex-disjoint edges $E^*\subseteq E$, where each edge $e\in E^*$ has demand $d(e)\in\mathbb{Z}_{\geq 1}$, and a vehicle capacity $k\in\mathbb{Z}_{\geq 1}$, the objective is to find a set of tours $\T$ in $G$ such that
\begin{enumerate}
    \item[(1)] each tour $T\in\T$ servers all customers in $E^*_T= E^*\cap E(T)$ with total demand at most $k$, i.e., $\sum_{e\in E^*_T}d(e)\leq k$;
    \item[(2)] all tours together serve all customers, i.e., $\bigcup_{T\in\T}E^*_T=E^*$;
    \item[(3)] each customer is served by exactly one tour in $\T$, i.e., $E^*_T\cap E^*_{T'}=\emptyset$ for all $T,T'\in\T$; 
    \item[(4)] the total cost, i.e., $\sum_{T\in\T}c(T)$, is minimized.
\end{enumerate}
\end{definition}

By our assumptions, we also have 
\begin{enumerate}
    \item[(5)] the tours in $\T^*$ are pairwise edge-disjoint, and any two tours share only the common vertex $v_0$, i.e., $E(T)\cap E(T')=\emptyset$ and $V_T\cap V_{T'}=\{v_0\}$ for all $T,T'\in\T$.
\end{enumerate}

We consider unite-demand CARP, where we have $d(e)=1$ for all $e\in E^*$.

\section{The Algorithms}
We first review the previous approximation algorithm given by  \citet{jansen1993bounds}. %Then, we propose our algorithm. 

\subsection{The previous tour partition algorithm}
For equal-demand CVRP, a well-known algorithm is the \emph{Iterated Tour Partitioning} (ITP) algorithm~\citep{HaimovichK85, altinkemer1990heuristics}.
ITP first computes an approximate Hamiltonian cycle and then obtains a solution to equal-demand CVRP by appropriately partitioning the cycle into fragments (paths), each of which is transformed into a tour.

For equal-demand CARP, \citet{jansen1993bounds} extended the ITP framework and proposed a similar tour partitioning algorithm, referred to as JITP.
In the first step of JITP, instead of computing an arbitrary Hamiltonian cycle, the algorithm constructs an RPP tour that traverses all edges in $E^*$.
Then, analogous to ITP, JITP partitions the RPP tour into a set of fragments, and each fragment is transformed into a tour by connecting its two endpoints to the depot.
JITP may explore multiple partitioning strategies and select the one that yields the minimum-cost CARP solution.

Specifically, let $m=\frac{n}{2}$, and suppose an RPP tour in $G$ is given by $(v_0,v_1,v_2,...,v_{2m-1},v_{2m},v_0)$, where $(v_{2i-1},v_{2i})\in E^*$ for each $i\in[m]$.
For convince, let $x_i=(v_{2i-1},v_{2i})$ for each $i\in[m]$, and use $(x_i,...,x_{i+i'})$ to represent the tour $(v_0,v_{2i-1},v_{2i},...,v_{2i+2i'-1},v_{2i+2i'},v_0)$, which is obtained by connecting the endpoints of the fragment $(v_{2i-1},v_{2i},...,v_{2i+2i'-1},v_{2i+2i'})$ to the depot. Then, IJTP returns the best solution among $k$ potential ones. 

In the $i$-th potential solution, denoted by $\T_{i}$, we have $\T_{i}=\{T^1_i,...,T^N_i\}$, where $N=\ceil{\frac{n-i}{k}}+1$, $T^1_i=(x_1,...,x_i)$, $T^j_i=(x_{i+(j-2)k+1},...,x_{i+(j-1)k})$ for each $2\leq j< N$, and $T^{N}_i=(x_{i+(N-2)k+1},...,x_{n})$.

Therefore, in the solution of JITP, except possibly for the first and the last tours, each tour serves exactly $k$ customers. We remark that \citet{wohlk2008decade} proved that the optimal partitioning can be computed in polynomial time using a dynamic programming approach.

We have the following result.

\begin{lemma}[\citep{jansen1993bounds,wohlk2008approximation,van2014constant}]\label{JITP}
For equal-demand CARP, given an RPP tour $H$ in $G$ with $E^*\subseteq E(H)$, there is a polynomial-time algorithm that computes a solution with cost at most $\frac{2}{\capacity}\delta(E^*)+\frac{\capacity-1}{\capacity}c(H)$.
\end{lemma}

\subsection{Our algorithm}
In our algorithm, we use JITP as a subroutine. We first construct two RPP tours, and then use the one with smaller cost to call the JITP algorithm. 
The first RPP tour, denoted by $H_1$, is obtained directly by using the $\frac{3}{2}$-approximation algorithm~\citep{eiselt1995arc}. 
The second one, denoted by $H_2$, is obtained in the following way.

First, we find a minimum-cost perfect matching $M$ in the graph $G[V]$~\citep{schrijver2003combinatorial}. %Note that  $E(M)\cup E^*$ forms a subgraph where each vertex in $V$ has an even degree. 
Then, in the multi-graph $G'=(V\cup\{v_0\}, E(M)\cup E^*)$, every vertex has an even degree. Moreover, the graph consists of a set of components.
Then, similar to Kruskal’s algorithm~\citep{schrijver2003combinatorial}, we connect these components by greedily adding minimum-cost edges between them. Let $F$ be the set of added edges. Then, adding all edges in $F$ to $G'$ results in a connected graph. Then, the multi-graph $G''=(V\cup\{v_0\}, E(M)\cup E^*\cup F\cup F)$ forms an \emph{Eulerian graph}, i.e., a connected graph and every vertex in it has an even degree. Thus, a tour $T$ such that $E(T)=E(G'')$, also known as an \emph{Eulerian tour}, can be obtained in polynomial time~\citep{schrijver2003combinatorial}. Then, by further shortcutting the tour $W$, an RPP tour is obtained.

We remark that the above algorithm uses the idea of the double-tree algorithm for TSP~\citep{williamson2011design}.

%$=(V\cup\{v_0\}, E(M)\cup E^*\cup F)$ forms a connected graph. Moreover, 

%$G'=(V,E\setminus E^*)$.

An illustration of our algorithm for equal-demand CARP can be found in Algorithm~\ref{alg}.

\begin{algorithm}[ht]
\caption{The approximation algorithm for equal-demand CARP}
\label{alg}
\small
\vspace*{2mm}
\textbf{Input:} An instance $G=(V,E)$ of equal-demand CARP. \\
\textbf{Output:} A solution to equal-demand CARP.

\begin{algorithmic}[1]
\State\label{line1} Obtain an RPP tour $H_1$ by using the $\frac{3}{2}$-approximation algorithm~\citep{eiselt1995arc}.

\State\label{line2} Find a minimum-cost perfect matching $M$ in $G[V]$~\citep{schrijver2003combinatorial}. 

\State\label{line3} Connect the components in $G'=(V\cup\{v_0\}, E(M)\cup E^*)$ by greedily adding minimum-cost edges between them, and let $F$ be the set of added edges.

\State\label{line4} Obtain an RPP tour $H_2$ by finding an Eluerian tour $T$ in $G''=(V\cup\{v_0\}, E(M)\cup E^*\cup F\cup F)$ and then shortcutting $T$.

\State\label{line5} Obtain a solution $\T$ by using the RPP tour $\arg\min\{c(H_1),c(H_2)\}$ to call the JITP algorithm.

\State \Return $\T$.
\end{algorithmic}
\end{algorithm}

\section{Performance Analysis}
In this section, we use $\T^*$ to denote an optimal solution to equal-demand CARP and define $\OPT=c(\T^*)$. 
In the following, we first recall two lower bounds on $\OPT$ in~\citep{jansen1993bounds}; then, we analyze the upper bounds on the cost of the used RPP tours in our algorithm and propose new lower bounds on $\OPT$; last, we analyze the approximation ratio of our algorithm by using these lower bounds on $\OPT$.

Recall that $H^*$ denotes an optimal RPP tour, and $\delta(E')=\sum_{(v_i,v_j)\in E'}\frac{c(v_0,v_i)+c(v_i,v_j)+c(v_0,v_j)}{2}$ for any $E'\subseteq E^*$. 
There are two known lower bounds on $\OPT$.

\begin{lemma}[\citep{jansen1993bounds}]\label{rrp}
It holds that $\frac{2}{k}\delta(E^*)\leq\OPT$ and $c(H^*)\leq\OPT$.
\end{lemma}

Next, we analyze the used RPP tours in our algorithm.

Let $E'$ denote a minimum-cost edge set such that $E'\cup E^*$ forms a spanning tree in $G$. Then, we define $\MST=c(E')+c(E^*)$.

The first RPP tour $H_1$ is obtained by using the $\frac{3}{2}$-approximation algorithm~\citep{eiselt1995arc}. We have the following result.

\begin{lemma}\label{rpp1}
It holds that $c(H_1)\leq\MST+\frac{1}{2}\OPT$.
\end{lemma}
\begin{proof}
It is well known that $c(H_1)\leq \MST+\frac{1}{2}c(H^*)$~\citep{eiselt1995arc}. Since $c(H^*)\leq \OPT$ by Lemma~\ref{rrp}, we obtain the desired result.
\end{proof}

\begin{lemma}\label{rpp2}
It holds that $c(H_2)\leq\OPT+2\MST-2c(E^*)$.
\end{lemma}
\begin{proof}
Recall that $H_2$ is obtained by shortcutting the graph $G''=(V\cup\{v_0\}, E(M)\cup E^*\cup F\cup F)$. By the triangle inequality, we have 
\begin{equation}\label{eq1}
c(H_2)\leq c(E^*)+c(M)+2c(F).
\end{equation}

By line~\ref{line2}, the graph $G'=(V\cup\{v_0\}, E(M)\cup E^*)$ is the minimum-cost Eulerian graph containing all edges in $E^*$. 
By shortcutting the depot $v_0$ from the optimal solution $\T^*$, we obtain an Eulerian graph containing all edges in $E^*$ with cost at most $\OPT$. Thus, we have 
\begin{equation}\label{eq2}
c(E^*)+c(M)\leq\OPT. 
\end{equation}

By line~\ref{line3} and the proof of Kruskal’s algorithm~\citep{schrijver2003combinatorial}, $F$ is the minimum-cost set of edges such that adding all edges in $F$ to the graph $G'=(V\cup\{v_0\}, E(M)\cup E^*)$ yields a connected graph. 
Recall that $E'$ is a minimum-cost edge set such that $E'\cup E^*$ forms a spanning tree in $G$. Thus, adding all edges in $E'$ to $G'$ yields a connected graph, and then we have 
\begin{equation}\label{eq3}
c(F)\leq c(E')=\MST-c(E^*).
\end{equation}

Therefore, by (\ref{eq1}), (\ref{eq2}), and (\ref{eq3}), we have $c(H_2)\leq\OPT+2\MST-2c(E^*)$.
\end{proof}

For any tour $T=(v_0,v_1,v_2,...,v_{2t-1},v_{2t},v_0)$ in the optimal solution $\T^*$, we have $t\leq k$. Recall that $E^*_T=\{(v_{2i-1},v_{2i})\mid i\in [t]\}$ and $V_T=\{v_0,v_1,...,v_{2t}\}$. 
We also define $MST_T$ be the cost of a minimum-cost spanning tree in the graph $G[V_T]$ containing all edges in $E^*_T$.

We have the following property.

\begin{lemma}\label{tem}
It holds that $\delta(E^*)=\sum_{T\in\T^*}\delta(E^*_T)$ and $\MST\leq\sum_{T\in\T^*}\MST_T$.
\end{lemma}
\begin{proof}
Since $\T^*$ is a solution to equal-demand CARP, we have $E^*=\bigcup_{T\in\T^*}E^*_T$ and $E^*_T\cap E^*_{T'}=\emptyset$ for all $T,T'\in\T$ by (2) and (3) in Definition~\ref{def}. Thus, we have $\delta(E^*)=\sum_{T\in\T^*}\delta(E^*_T)$.

Moreover, by (5) in Definition~\ref{def}, we have $V=\bigcup_{T\in\T^*}V_T$ and $V_{T}\cap V_{T'}=\{v_0\}$ for all $T,T'\in\T^*$. Since $MST_T$ measures the cost of a spanning tree in $G[V_T]$, then $\sum_{T\in\T^*}\MST_T$ measures the cost of a spanning tree in $G$. By definition, any spanning tree in $G$ containing all edges in $E^*_T$ has cost at least $\MST$. Thus, we have $\MST\leq\sum_{T\in\T^*}\MST_T$.
\end{proof}

By extending the techniques in~\citep{DBLP:conf/mfcs/zhao}, we obtain the following key result.

\begin{lemma}\label{thebounds}
For any tour $T\in\T^*$, there exist parameters $\{\alpha_1,\alpha_2,...,\alpha_m\}$ such that
\begin{enumerate}
    \item[(1)] $\delta(E^*_T)\leq\lrA{\frac{\capacity}{2}+\alpha-\sum_{i=1}^{m}i\alpha_i}c(T)$;
    \item[(2)] $\MST_T\leq\lrA{1-\max_{\substack{1\leq i\leq m}}\frac{1}{2}\alpha_{i}}c(T)$;
    \item[(3)] $c(E^*_T)= (1-\alpha)c(T)$;
    \item[(4)] $\sum^m_{i=1}\alpha_i=\alpha\leq 1$, where $m=\infty$ and  $0\leq\alpha_i\leq \alpha$ for each $1\leq i\leq m$.
\end{enumerate}
\end{lemma}
\begin{proof}
To avoid distraction from our main discussions, we delay the proof to Section~\ref{theproof}.
\end{proof}

Next, we are ready to analyze the approximation ratio of our algorithm.
%We have the following results.
\begin{lemma}\label{opt}
The approximation ratio of our algorithm is at most
\[
\max_{\substack{l\in\mathbb{Z}_{\geq1}\\1\geq\alpha\geq0}}\min\lrC{\tau(\alpha,l),\eta(\alpha,l)},
\]
where $\tau(\alpha,l)=\frac{5\capacity-3}{2\capacity}-\frac{(2l^2-2l+k-1)\alpha}{2kl}$ and $\eta(\alpha,l)=\frac{2\capacity-1}{\capacity}-\frac{(l^2+l-2kl+k-1)\alpha}{kl}$.
% \[
% \tau(\alpha,l)=\frac{2\capacity-1}{\capacity}-\frac{(2l^2-2l+k-1)\alpha}{2kl}\quad\text{and}\quad\eta(\alpha,l)=\frac{2\capacity-1}{\capacity}-\frac{(l^2+l-2kl+k-1)\alpha}{kl}.
% \]
\end{lemma}
\begin{proof}
Recall that our algorithm uses the RPP tour $\arg\min\{c(H_1),c(H_2)\}$ to call the JITP algorithm. Then, by Lemma~\ref{JITP}, the obtained solution, denoted by $\T$, satisfies that
\begin{equation}\label{upperbound}
\begin{split}
c(\T)&\leq \frac{2}{\capacity}\delta(E^*)+\frac{\capacity-1}{\capacity}\min\{c(H_1),c(H_2)\}\\
&\leq \frac{2}{\capacity}\delta(E^*)+\frac{\capacity-1}{\capacity}\min\lrC{\MST+\frac{1}{2}\OPT,\OPT+2\MST-2c(E^*)}\\
&\leq\sum_{T\in\T^*}\lrA{\frac{2}{\capacity}\delta(E^*_T)+\frac{\capacity-1}{\capacity}\min\lrC{\MST_T+\frac{1}{2}c(T),c(T)+2\MST_T-2c(E^*_T)}},
\end{split}
\end{equation}
where the second inequality follows from Lemmas~\ref{rpp1} and \ref{rpp2}, and the last inequality from Lemma~\ref{tem} and the fact that $\OPT=\sum_{T\in\T^*}c(T)$.

By Lemma~\ref{thebounds}, we have
\begin{equation}\label{userpp1}
\begin{split}
&\frac{2}{\capacity}\delta(E^*_T)+\frac{\capacity-1}{\capacity}\lrA{\MST_T+\frac{1}{2}c(T)}\\
&\leq\frac{2}{\capacity}\lrA{\frac{\capacity}{2}+\alpha-\sum_{i=1}^{m}i\alpha_i}c(T)+\frac{\capacity-1}{\capacity}\lrA{1-\max_{\substack{1\leq i\leq m}}\frac{1}{2}\alpha_{i}+\frac{1}{2}}c(T)\\
&=\frac{5\capacity-3}{2\capacity}c(T)+\frac{2}{\capacity}\lrA{\alpha-\sum_{i=1}^{m}i\alpha_i}c(T)-\lrA{\max_{\substack{1\leq i\leq m}}\frac{\capacity-1}{2\capacity}\alpha_{i}}c(T),
\end{split}
\end{equation}
and
\begin{equation}\label{userpp2}
\begin{split}
&\frac{2}{\capacity}\delta(E^*_T)+\frac{\capacity-1}{\capacity}\lrA{c(T)+2\MST_T-2c(E^*_T)}\\
&\leq\frac{2}{\capacity}\lrA{\frac{\capacity}{2}+\alpha-\sum_{i=1}^{m}i\alpha_i}c(T)+\frac{\capacity-1}{\capacity}\lrA{1+2-\max_{\substack{1\leq i\leq m}}\alpha_{i}-2(1-\alpha)}c(T)\\
&=\frac{2\capacity-1}{\capacity}c(T)+\frac{2}{\capacity}\lrA{\alpha-\sum_{i=1}^{m}i\alpha_i}c(T)+\frac{\capacity-1}{\capacity}\lrA{2\alpha-\max_{\substack{1\leq i\leq m}}\alpha_{i}}c(T)\\
&=\frac{2\capacity-1}{\capacity}c(T)+\lrA{2\alpha-\frac{2}{\capacity}\sum_{i=1}^{m}i\alpha_i}c(T)-\lrA{\max_{\substack{1\leq i\leq m}}\frac{\capacity-1}{\capacity}\alpha_{i}}c(T).
\end{split}
\end{equation}

Let
\begin{equation}\label{thef}
f(\vec{\alpha})=\frac{5\capacity-3}{2\capacity}+\frac{2}{\capacity}\lrA{\alpha-\sum_{i=1}^{m}i\alpha_i}-\lrA{\max_{\substack{1\leq i\leq m}}\frac{\capacity-1}{2\capacity}\alpha_{i}},
\end{equation}
and
\begin{equation}\label{theg}
g(\vec{\alpha})=\frac{2\capacity-1}{\capacity}+\lrA{2\alpha-\frac{2}{\capacity}\sum_{i=1}^{m}i\alpha_i}-\lrA{\max_{\substack{1\leq i\leq m}}\frac{\capacity-1}{\capacity}\alpha_{i}},
\end{equation}
where $\vec{\alpha}=(\alpha_1,\alpha_2,...,\alpha_m)$.

By (\ref{upperbound}), (\ref{userpp1}), (\ref{userpp2}), and  Lemma~\ref{thebounds}, the approximation ratio of our algorithm is at most
\begin{equation}\label{iratio1}
\max_{\substack{\alpha_1+\alpha_2+\cdots+\alpha_m = \alpha\\
\alpha_1,\alpha_2,\dots,\alpha_m\geq 0\\1\geq\alpha\geq0}}\min\lrC{f(\vec{\alpha}),g(\vec{\alpha})}.
\end{equation}

Under the three conditions (or constraints) in (\ref{iratio1}), both $f(\vec{\alpha})$ and $g(\vec{\alpha})$ achieve their maximum values only when $\alpha_1\geq\alpha_2\geq\cdots\geq\alpha_m$. As shown in~\citep{DBLP:conf/mfcs/zhao}, the reason is that if there exists $\alpha_{p}<\alpha_{q}$ for some $p<q$, then exchanging their values yields a larger solution: the value $\max_{\substack{1\leq i\leq m}}\alpha_{i}$ remain unchanged, while the coefficients of $\alpha_p$ and $\alpha_q$ satisfy $0>\frac{-2p}{k}>\frac{-2q}{k}$.

Therefore, the approximation ratio of our algorithm is at most
\begin{equation}\label{iratio2}
\begin{split}
\max_{\substack{\alpha_1+\alpha_2+\cdots+\alpha_m = \alpha\\
\alpha_1\geq\alpha_2\geq\cdots\geq\alpha_m\geq 0\\1\geq\alpha\geq0}}\min\lrC{f(\vec{\alpha}),g(\vec{\alpha})}.
\end{split}
\end{equation}

Fixing the value of $\alpha_1$, and under the three conditions in (\ref{iratio2}), maximizing $f(\vec{\alpha})$ (or $g(\vec{\alpha})$) is equivalent to minimizing $\sum_{i=1}^{m}i\alpha_i$. 
It is easy to see that each $\alpha_i$ should be set as large as possible, in order of increasing index. 
Specifically, when $\alpha-\sum_{j=1}^{i-1}\alpha_j\geq 0$, set $\alpha_i=\min\{\alpha-\sum_{j=1}^{i-1}\alpha_j,\alpha_1\}$. Then, we have $\alpha_1=\alpha_2=\cdots=\alpha_{l-1}\geq \alpha_l=\alpha-(l-1)\alpha_1\geq 0$ and $\alpha_{l+1}=\alpha_{l+2}=\cdots =\alpha_m=0$, where $l=\ceil{\frac{\alpha}{\alpha_1}}$. Therefore, we have $\sum_{i=1}^{m}i\alpha_i=\sum_{i=1}^{l-1}i\alpha_i+l\alpha_l=\frac{(l-1)l\alpha_1}{2}+l(\alpha-(l-1)\alpha_1)=l\alpha-\frac{(l-1)l\alpha_1}{2}$.

Let
\[
z(\alpha_1)=-l\alpha+\frac{(l-1)l\alpha_1}{2}-C_1\cdot\alpha_1,
\]
where $C_1>0$ denotes a constant.

When $\alpha_1\in[\frac{\alpha}{t+1}, \frac{\alpha}{t})$ with $t\in\mathbb{Z}_{\geq1}$, the function $z(\alpha_1)$ is a continuous linear function. Moreover, since $l=\ceil{\frac{\alpha}{\alpha_1}}=t+1$, we have
\begin{equation}\label{leftlim}
\lim_{\alpha_1\to (\frac{\alpha}{t})^-}z(\alpha_1)=-(t+1)\alpha+\frac{t(t+1)\cdot\frac{\alpha}{t}}{2}-C_1\cdot\frac{\alpha}{t}=-\frac{(t+1)\alpha}{2}-C_1\cdot\frac{\alpha}{t}.
\end{equation}
When $\alpha_1=\frac{\alpha}{t}$, since $l=\ceil{\frac{\alpha}{\alpha_1}}=t$, we have
\begin{equation}\label{lim}
z(\frac{\alpha}{t})=-t\alpha+\frac{(t-1)\alpha}{2}-C_1\cdot\frac{\alpha}{t}=-\frac{(t+1)\alpha}{2}-C_1\cdot\frac{\alpha}{t}.
\end{equation}

Therefore, by (\ref{leftlim}) and (\ref{lim}), the function $z(\alpha_1)$ is a continuous function when $\alpha_1\in(0,\alpha]$. Recall that $z(\alpha_1)$ is a continuous linear function when $\alpha_1\in[\frac{\alpha}{t+1}, \frac{\alpha}{t})$ with $t\in\mathbb{Z}_{\geq1}$. So, it is easy to observe that there exists at least one value of $\alpha_1$ maximizing $z(\alpha_1)$ such that $\frac{\alpha}{\alpha_1}$ is an integer.

Consequently, there exists at least one value of $\alpha_1$ maximizing $f(\vec{\alpha})$ (or $g(\vec{\alpha})$) such that $\frac{\alpha}{\alpha_1}$ is an integer.

Under $\frac{\alpha}{\alpha_1}=l\in\mathbb{Z}_{\geq1}$ and $\alpha_1=\alpha_2=\cdots=\alpha_l=\frac{\alpha}{\alpha_1}=\frac{\alpha}{l}$, as shown before, we have $\sum_{i=1}^{m}i\alpha_i=l\alpha-\frac{(l-1)l\alpha_1}{2}=l\alpha-\frac{(l-1)\alpha}{2}=\frac{(l+1)\alpha}{2}$. Then, by (\ref{thef}) and (\ref{theg}), we have
\begin{align*}
f(\vec{\alpha})&=\frac{5\capacity-3}{2\capacity}+\frac{2}{\capacity}\lrA{\alpha-\sum_{i=1}^{m}i\alpha_i}-\lrA{\max_{\substack{1\leq i\leq m}}\frac{\capacity-1}{2\capacity}\alpha_{i}}\\
&=\frac{5\capacity-3}{2\capacity}+\frac{2}{\capacity}\lrA{\alpha-\frac{(l+1)\alpha}{2}}-\frac{\capacity-1}{2\capacity}\cdot\frac{\alpha}{l}\\
&=\frac{5\capacity-3}{2\capacity}-\frac{(2l^2-2l+k-1)\alpha}{2kl},
\end{align*}
and
\begin{align*}
g(\vec{\alpha})&=\frac{2\capacity-1}{\capacity}+\lrA{2\alpha-\frac{2}{\capacity}\sum_{i=1}^{m}i\alpha_i}-\lrA{\max_{\substack{1\leq i\leq m}}\frac{\capacity-1}{\capacity}\alpha_{i}}\\
&=\frac{2\capacity-1}{\capacity}+\lrA{2\alpha-\frac{2}{\capacity}\cdot\frac{(l+1)\alpha}{2}}-\frac{\capacity-1}{\capacity}\cdot\frac{\alpha}{l}\\
&=\frac{2\capacity-1}{\capacity}-\frac{(k-1)\alpha+l(l+1)\alpha-2kl\alpha}{kl}\\
&=\frac{2\capacity-1}{\capacity}-\frac{(l^2+l-2kl+k-1)\alpha}{kl}.
\end{align*}

Let
\[
\tau(\alpha,l)=\frac{5\capacity-3}{2\capacity}-\frac{(2l^2-2l+k-1)\alpha}{2kl}\quad\text{and}\quad\eta(\alpha,l)=\frac{2\capacity-1}{\capacity}-\frac{(l^2+l-2kl+k-1)\alpha}{kl}.
\]
% and
% \[
% \eta(\alpha,l)=\frac{2\capacity-1}{\capacity}-\frac{(l^2+l-2kl+k-1)\alpha}{kl}.
% \]

Then, by (\ref{iratio2}), the approximation ratio of our algorithm is at most
\begin{equation}\label{iratio3}
\begin{split}
\max_{\substack{\alpha_1+\alpha_2+\cdots+\alpha_m = \alpha\\
\alpha_1\geq\alpha_2\geq\cdots\geq\alpha_m\geq 0\\1\geq\alpha\geq0}}\min\lrC{f(\vec{\alpha}),g(\vec{\alpha})}=\max_{\substack{l\in\mathbb{Z}_{\geq1}\\1\geq\alpha\geq0}}\min\lrC{\tau(\alpha,l),\eta(\alpha,l)},
\end{split}
\end{equation}
as desired.
\end{proof}

\begin{theorem}\label{mainres}
For unit-demand CARP, there is a polynomial-time algorithm with an approximation ratio of $\frac{5}{2}-\frac{2l^2+10l+k-4}{2k(4l-1)}$, where $l=\ceil{\frac{\sqrt{8k-7}-1}{4}}$.
\end{theorem}
\begin{proof}
It is clear that our algorithm takes $O(n^3)$ time.

Next, we compute the approximation ratio given by Lemma~\ref{opt}.

Recall that $\tau(\alpha,l)=\frac{5\capacity-3}{2\capacity}-\frac{(2l^2-2l+k-1)\alpha}{2kl}$ and $\eta(\alpha,l)=\frac{2\capacity-1}{\capacity}-\frac{(l^2+l-2kl+k-1)\alpha}{kl}$. Therefore, when $\alpha\leq\frac{l}{4l-1}$, we have $\eta(\alpha,l)\leq \tau(\alpha,l)$; otherwise, we have $\eta(\alpha,l)\geq \tau(\alpha,l)$.

\textbf{Case~1: $0\leq\alpha\leq\frac{l}{4l-1}$.} In this case, by Lemma~\ref{opt}, the approximation ratio of our algorithm is at most $\max_{\substack{l\in\mathbb{Z}_{\geq1},\frac{l}{4l-1}\geq\alpha\geq0}}\eta(\alpha,l)$. Since $\eta(\alpha,l)$ is a linear function w.r.t.\ $\alpha$, we have
\begin{equation}\label{1case1}
\max_{\substack{l\in\mathbb{Z}_{\geq1}\\ \frac{l}{4l-1}\geq\alpha\geq0}}\eta(\alpha,l)=\max_{l\in\mathbb{Z}_{\geq1}}\max\lrC{\eta(0,l),\eta(\frac{l}{4l-1},l)}.
\end{equation}

We have $\eta(0,l)=\frac{2\capacity-1}{\capacity}$. Next, we consider $\max_{l\in\mathbb{Z}_{\geq1}}\eta(\frac{l}{4l-1},l)$. 

First, we have
\begin{equation}\label{1case2}
\eta(\frac{l}{4l-1},l)=\frac{2\capacity-1}{\capacity}-\frac{l^2+l-2\capacity l+\capacity-1}{\capacity(4l-1)}.
\end{equation}
Then, we obtain
\begin{equation}\label{1case3}
\begin{split}
\frac{\mathrm{d}}{\mathrm{d}l}{\eta(\frac{l}{4l-1},l)}&=\frac{ -\capacity(2l + 1 - 2\capacity)(4l - 1) + 4\capacity(l^2 + l - 2\capacity l + \capacity - 1) }{\capacity^2(4l - 1)^2}=\frac{-4l^2+2l+2\capacity-3}{\capacity(4l - 1)^2}.
\end{split}
\end{equation}

Therefore, by (\ref{1case3}), when $\capacity\geq3$ and $l\geq1$, $\eta(\frac{l}{4l-1},l)$ is a concave function w.r.t.\ $l$. 

Consequently, we have 
\begin{equation}\label{1case4}
\max_{l\in\mathbb{Z}_{\geq1}}\eta(\frac{l}{4l-1},l)=\eta(\frac{\widetilde{l}}{4\widetilde{l}-1},\widetilde{l}),
\end{equation}
where 
\begin{equation}\label{1case5}
\widetilde{l}=\arg\min_{l\in\mathbb{Z}_{\geq1}}\lrC{\eta(\frac{l}{4l-1},l)-\eta(\frac{(l+1)}{4(l+1)-1},(l+1))\geq 0}.
\end{equation}

Note that
\begin{equation}\label{1case6}
\begin{split}
&\eta(\frac{l}{4l-1},l)-\eta(\frac{(l+1)}{4(l+1)-1},(l+1))\\
&=\lrA{\frac{(l+1)^2+(l+1)-2\capacity (l+1)+\capacity-1}{\capacity(4l+3)}}-\lrA{\frac{l^2+l-2\capacity l+\capacity-1}{\capacity(4l-1)}}\\
&=\frac{((l+1)^2+(l+1)-2\capacity (l+1)+\capacity-1)(4l-1)-(l^2+l-2\capacity l+\capacity-1)(4l+3)}{\capacity(4l-1)(4l+3)}\\
% &=\frac{(l^2+3l-2\capacity l+1-\capacity)(4l-1)-(l^2+l-2\capacity l+\capacity-1)(4l+3)}{\capacity(4l-1)(4l+3)}\\
% &=\frac{(4l^3+11l^2-8\capacity l^2+l-2\capacity l+\capacity-1)-(4l^3+7l^2-8\capacity l^2-2\capacity l-l + 3\capacity-3)}{\capacity(4l-1)(4l+3)}\\
&=\frac{2(2l^2+l-\capacity+1)}{\capacity(4l-1)(4l+3)}.
\end{split}
\end{equation}

By (\ref{1case6}), when $\eta(\frac{l}{4l-1},l)-\eta(\frac{(l+1)}{4(l+1)-1},(l+1))=0$ and $l\geq 1$, we have $l=\frac{\sqrt{8k-7}-1}{4}$. Then, by (\ref{1case5}), we have $\widetilde{l}=\ceil{\frac{\sqrt{8k-7}-1}{4}}$.
Then, by (\ref{1case2}) and (\ref{1case4}), we have 
\begin{equation}\label{1case7}
\max_{l\in\mathbb{Z}_{\geq1}}\eta(\frac{l}{4l-1},l)=\frac{2\capacity-1}{\capacity}-\frac{\widetilde{l}^2+\widetilde{l}-2\capacity \widetilde{l}+\capacity-1}{\capacity(4\widetilde{l}-1)},\quad\text{where}\quad \widetilde{l}=\Ceil{\frac{\sqrt{8k-7}-1}{4}}.
\end{equation}

Recall that $\max_{l\in\mathbb{Z}_{\geq1}}\eta(0,l)=\frac{2\capacity-1}{\capacity}$. It can be verified that $\frac{2\capacity-1}{\capacity}<\frac{2\capacity-1}{\capacity}-\frac{\widetilde{l}^2+\widetilde{l}-2\capacity \widetilde{l}+\capacity-1}{\capacity(4\widetilde{l}-1)}$ for all $\capacity\geq3$. 
Therefore, in the case where $0\leq\alpha\leq\frac{l}{4l-1}$, by (\ref{1case1}), the approximation ratio of our algorithm is at most $\frac{2\capacity-1}{\capacity}-\frac{\widetilde{l}^2+\widetilde{l}-2\capacity \widetilde{l}+\capacity-1}{\capacity(4\widetilde{l}-1)}$, where $\widetilde{l}=\ceil{\frac{\sqrt{8k-7}-1}{4}}$.

\textbf{Case~2: $1\geq\alpha\geq \frac{l}{4l-1}$.} In this case, by Lemma~\ref{opt}, the approximation ratio of our algorithm is at most $\max_{\substack{l\in\mathbb{Z}_{\geq1},1\geq\alpha\geq\frac{l}{4l-1}}}\tau(\alpha,l)$. Since $\tau(\alpha,l)$ is a linear function w.r.t.\ $\alpha$, we have
\begin{equation}\label{2case1}
\max_{\substack{l\in\mathbb{Z}_{\geq1}\\1\geq\alpha\geq\frac{l}{4l-1}}}\tau(\alpha,l)=\max_{l\in\mathbb{Z}_{\geq1}}\max\lrC{\tau(\frac{l}{4l-1},l),\tau(1,l)}.
\end{equation}

Since $\tau(\frac{l}{4l-1},l)=\eta(\frac{l}{4l-1},l)$, by the previous analysis, we have
\begin{equation}\label{1case11}
\max_{l\in\mathbb{Z}_{\geq1}}\tau(\frac{l}{4l-1},l)=\frac{2\capacity-1}{\capacity}-\frac{\widetilde{l}^2+\widetilde{l}-2\capacity \widetilde{l}+\capacity-1}{\capacity(4\widetilde{l}-1)},\quad\text{where}\quad \widetilde{l}=\Ceil{\frac{\sqrt{8k-7}-1}{4}}.
\end{equation}

Next, we consider $\max_{l\in\mathbb{Z}_{\geq1}}\tau(1,l)$. 

First, we have
\begin{equation}\label{2case2}
\tau(1,l)=\frac{5\capacity-3}{2\capacity}-\frac{2l^2-2l+k-1}{2kl}.
\end{equation}
Then, we obtain
\begin{equation}\label{2case3}
\begin{split}
\frac{\mathrm{d}}{\mathrm{d}l}{\tau(1,l)}&=\frac{ -2\capacity l(4l - 2) + 2\capacity(2l^2 - 2l + \capacity - 1) }{4\capacity^2l^2}=\frac{-2l^2+\capacity-1}{2\capacity^2l^2}.
\end{split}
\end{equation}

Therefore, by (\ref{2case3}), when $\capacity\geq3$ and $l\geq1$, $\tau(1,l)$ is a concave function w.r.t.\ $l$. 

Consequently, we have 
\begin{equation}\label{2case4}
\max_{l\in\mathbb{Z}_{\geq1}}\tau(1,l)=\tau(1,\widehat{l}),
\end{equation}
where 
\begin{equation}\label{2case5}
\widehat{l}=\arg\min_{l\in\mathbb{Z}_{\geq1}}\lrC{\tau(1,l)-\tau(1,(l+1))\geq 0}.
\end{equation}

Note that
\begin{equation}\label{2case6}
\begin{split}
&\tau(1,l)-\tau(1,(l+1))\\
&=\lrA{\frac{2(l+1)^2-2(l+1)+\capacity-1}{2\capacity(l+1)}}-\lrA{\frac{2l^2-2l+k-1}{2kl}}\\
&=\frac{(2(l+1)^2-2(l+1)+\capacity-1)l-(2l^2-2l+k-1)(l+1)}{2\capacity l(l+1)}\\
&=\frac{(2l^2+2l+\capacity-1)l-(2l^2-2l+k-1)(l+1)}{2\capacity l(l+1)}\\
&=\frac{(2l^3+2l^2+\capacity l-l)-(2l^3 +\capacity l-3l+k-1)}{2\capacity l(l+1)}\\
&=\frac{2l^2+2l-\capacity+1}{2\capacity l(l+1)}.
\end{split}
\end{equation}

By (\ref{2case6}), when $\tau(1,l)-\tau(1,(l+1))=0$ and $l\geq 1$, we have $l=\frac{\sqrt{2k-1}-1}{2}$. Then, by (\ref{2case5}), we have $\widehat{l}=\ceil{\frac{\sqrt{2k-1}-1}{2}}$.
Then, by (\ref{2case2}) and (\ref{2case4}), we have 
\begin{equation}\label{2case7}
\max_{l\in\mathbb{Z}_{\geq1}}\tau(1,l)=\frac{5\capacity-3}{2\capacity}-\frac{2\widehat{l}^2-2\widehat{l}+k-1}{2k\widehat{l}},\quad\text{where}\quad \widehat{l}=\Ceil{\frac{\sqrt{2k-1}-1}{2}}.
\end{equation}

Recall that $\max_{l\in\mathbb{Z}_{\geq1}}\tau(\frac{l}{4l-1},l)=\frac{2\capacity-1}{\capacity}-\frac{\widetilde{l}^2+\widetilde{l}-2\capacity \widetilde{l}+\capacity-1}{\capacity(4\widetilde{l}-1)}$, where $\widetilde{l}=\ceil{\frac{\sqrt{8k-7}-1}{4}}$. 
It can be verified that $\frac{5\capacity-3}{2\capacity}-\frac{2\widehat{l}^2-2\widehat{l}+k-1}{2k\widehat{l}}<\frac{2\capacity-1}{\capacity}-\frac{\widetilde{l}^2+\widetilde{l}-2\capacity \widetilde{l}+\capacity-1}{\capacity(4\widetilde{l}-1)}$ for all $\capacity\geq 3$.
Thus, in the case where $1\geq\alpha\geq\frac{l}{4l-1}$, by (\ref{2case1}), the approximation ratio of our algorithm is at most $\frac{2\capacity-1}{\capacity}-\frac{\widetilde{l}^2+\widetilde{l}-2\capacity \widetilde{l}+\capacity-1}{\capacity(4\widetilde{l}-1)}$, where $\widetilde{l}=\ceil{\frac{\sqrt{8k-7}-1}{4}}$.

In conclusion, the approximation ratio of our algorithm is at most
\[
\frac{2\capacity-1}{\capacity}-\frac{\widetilde{l}^2+\widetilde{l}-2\capacity \widetilde{l}+\capacity-1}{\capacity(4\widetilde{l}-1)}=\frac{5}{2}-\frac{2\widetilde{l}^2+10\widetilde{l}+\capacity-4}{2\capacity(4\widetilde{l}-1)}=\frac{5}{2}-\Theta\lrA{\frac{1}{\sqrt{k}}},
\]
where $\widetilde{l}=\ceil{\frac{\sqrt{8k-7}-1}{4}}$.
\end{proof}

The approximation ratio in Theorem~\ref{mainres} is derived using two RPP tours $H_1$ and $H_2$ in $G$. Interestingly, it can be shown that the approximation ratio can be achieved using only the RPP tour $H_1$ in $G$ (see Section~\ref{onlyone}).

\begin{remark}
CVRP can be viewed as a special case of CARP by replacing each customer with a zero-cost edge. Thus, CARP with $\alpha=1$ capture the CVRP.
In fact, in the proof of Theorem~\ref{mainres}, the ratio 
%$\frac{5\capacity-3}{2\capacity}-\frac{2\widehat{l}^2-2\widehat{l}+k-1}{2k\widehat{l}}$, where $\widehat{l}=\Ceil{\frac{\sqrt{2k-1}-1}{2}}$,
\[
\frac{5\capacity-3}{2\capacity}-\frac{2\widehat{l}^2-2\widehat{l}+k-1}{2k\widehat{l}},\quad\text{where}\quad \widehat{l}=\Ceil{\frac{\sqrt{2k-1}-1}{2}},
\]
obtained in the case $\alpha=1$, exactly matches the approximation ratio for CVRP given in~\citep{DBLP:conf/mfcs/zhao}.
\end{remark}

\section{Proof of Lemma~\ref{thebounds}}\label{theproof}
\begingroup
\def\thelemma{\ref{thebounds}}
\begin{lemma}
For any tour $T\in\T^*$, there exist parameters $\{\alpha_1,\alpha_2,...,\alpha_m\}$ such that
\begin{enumerate}
    \item[(1)] $\delta(E^*_T)\leq\lrA{\frac{\capacity}{2}+\alpha-\sum_{i=1}^{m}i\alpha_i}c(T)$;
    \item[(2)] $\MST_T\leq\lrA{1-\max_{\substack{1\leq i\leq m}}\frac{1}{2}\alpha_{i}}c(T)$;
    \item[(3)] $c(E^*_T)= (1-\alpha)c(T)$;
    \item[(4)] $\sum^m_{i=1}\alpha_i=\alpha\leq 1$, where $m=\infty$ and  $0\leq\alpha_i\leq \alpha$ for each $1\leq i\leq m$.
\end{enumerate}
\end{lemma}
\endgroup
\begin{proof}
We assume w.l.o.g.\ that $c(T)>0$; otherwise, the lemma holds trivially.

Let $T=(v_0,v_1,v_2,...,v_{2l-1},v_{2l},v_0)$ and $t=\ceil{\frac{l+1}{2}}$. Then, we have $l\leq k$, and $t\leq m$. Note that we also have $E^*_T=\{(v_{2i-1},v_{2i})\mid i\in[l]\}$ and $\size{E^*_T}=l$.

To construct the parameters $\{\alpha_1,\alpha_2,...,\alpha_m\}$, we set $\alpha_i=0$ for all $t<i\leq m$. Then, it remains to show how to set $\alpha_i$ for each $i\in[t]$. 
We consider two cases based on the parity of $l$, i.e., $\size{E^*_T}$.

\textbf{Case~1: $\size{E^*_T}$ is odd.}
In this case, we have $t=\ceil{\frac{l+1}{2}}=\frac{l+1}{2}$.
We let
\[
\alpha_{i}=\frac{c(v_{2i-2},v_{2i-1})+c(v_{2l+2-2i},v_{(2l+3-2i)\bmod (2l+1)})}{c(T)},\quad\forall i\in[t],
\]
and
\[
\beta_{i}=\frac{c(v_{2i-1},v_{2i})+c(v_{2l+1-2i},v_{2l+2-2i})}{c(T)},\quad\forall i\in [t-1].
\]
Also, we let $\beta_t=\frac{c(v_{2t-1},v_{2t})}{c(T)}$. Then, we have $\alpha_i,\beta_i\geq0$ for all $i\in[t]$.

Note that $c(E(T)\setminus E^*_T)=\sum_{i=1}^t\alpha_ic(T)$ and $c(E^*_T)=\sum_{i=1}^t\beta_ic(T)$.
Then, $\sum_{i=1}^{t}(\alpha_i+\beta_i)=1$. 
Since $\alpha_i=0$ for any $t<i\leq m$, we have $\alpha=\sum^m_{i=1}\alpha_i=\sum_{i=1}^{t}\alpha_i=1-\sum_{i=1}^{t}\beta_i\leq 1$. 

Therefore, the properties (3) and (4) in Lemma~\ref{thebounds} are both satisfied. 

Next, we prove the properties (1) and (2) in Lemma~\ref{thebounds}.

\begin{claim}\label{claim1}
When $\size{E^*_T}$ is odd, it holds that $\MST_T\leq\lrA{1-\max_{\substack{1\leq i\leq m}}\frac{1}{2}\alpha_{i}}c(T)$.
\end{claim}
\begin{proof}[Claim Proof]
Recall that $\MST_T$ measures the cost of a minimum-cost spanning tree in $G[V_T]$ containing all edges in $E^*_T$.
Deleting an arbitrary edge in $E(T)\setminus E^*_T$ from $T$ forms a spanning tree in $G[V_T]$ containing all edges in $E^*_T$. Then, we have 
\begin{align*}
\MST_T&\leq c(T)-\max_{e\in E(T)\setminus E^*_T}c(e)\\
&= c(T)-\max_{1\leq i\leq 2t}c(v_{2i-2},v_{(2i-1)\bmod (2l+1)})\\
&\leq c(T)-\max_{\substack{1\leq i\leq t}}\frac{1}{2}[c(v_{2i-2},v_{2i-1})+c(v_{2l+2-2i},v_{(2l+3-2i)\bmod (2l+1)})]\\
&=\lrA{1-\max_{\substack{1\leq i\leq t}}\frac{1}{2}\alpha_{i}}c(T)\\
&=\lrA{1-\max_{\substack{1\leq i\leq m}}\frac{1}{2}\alpha_{i}}c(T),
\end{align*}
where the first equality follows from $E(T)\setminus E^*_T=\{(v_{2i-2},v_{(2i-1)\bmod (2l+1)})\mid i\in[2t]\}$, the second inequality from $E(T)\setminus E^*_T=\{(v_{2i-2},v_{2i-1}),(v_{2i-2},v_{(2i-1)\bmod (2l+1)})\mid i\in[t]\}$, and the last inequality from the definition of $\alpha_i$.
\end{proof}

\begin{claim}\label{claim2}
When $\size{E^*_T}$ is odd, it holds that $\delta(E^*_T)\leq\lrA{\frac{\capacity}{2}+\alpha-\sum_{i=1}^{m}i\alpha_i}c(T)$.
\end{claim}
\begin{proof}[Claim Proof]
Recall that $\delta(E^*_T)=\sum_{(x,y)\in E^*_T}\frac{c(v_0,x)+c(x,y)+c(v_0,y)}{2}$. Since $E^*_T=\{(v_{2i-1},v_{2i})\mid i\in[2t-1]\}$, we have 
\begin{equation}\label{eq4}
2\delta(E^*_T)=c(E^*_T)+\sum_{i=1}^{2t-1}[{c(v_0, v_{2i-1})+c(v_0, v_{2i})}].%=c(E^*_T)+\sum_{i=1}^{2t-1}\lrA{c(v_0,v_i)+c(v_0,v_{2l+1-i})}.
\end{equation}

Since $E^*_T=\{(v_{2i-1},v_{2i}), (v_{2l+1-2i},v_{2l+2-2i})\mid i\in[t-1]\}\cup\{(v_{2t-1},v_{2t})\}$, we have
\begin{equation}\label{eq5}
\begin{split}
&\sum_{i=1}^{2t-1}\lrB{c(v_0, v_{2i-1})+c(v_0, v_{2i})}\\
&=c(v_0,v_{2t-1})+c(v_0,v_{2t})+\sum_{i=1}^{t-1}\lrB{c(v_0, v_{2i-1})+c(v_0, v_{2i})+c(v_0,v_{2l+1-2i})+c(v_0,v_{2l+2-2i})}\\
&\leq c(v_0,v_{2t-1})+c(v_0,v_{2t})+\sum_{i=1}^{t-1}\lrB{2c(v_0, v_{2i-1})+2c(v_0,v_{2l+2-2i})+\beta_ic(T)},
\end{split}
\end{equation}
where the inequality follows from $c(v_0, v_{2i})\leq c(v_0, v_{2i-1})+c(v_{2i-1},v_{2i})$ and $c(v_0,v_{2l+1-2i})\leq c(v_0,v_{2l+2-2i})+c(v_{2l+1-2i},v_{2l+2-2i})$ by the triangle inequality and $\beta_ic(T)=\frac{c(v_{2i-1},v_{2i})+c(v_{2l+1-2i},v_{2l+2-2i})}{c(T)}$ when $i\in[t-1]$.

For any $i\in[t]$, by the triangle inequality, we have
\begin{equation}\label{eq6}
c(v_0, v_{2i-1})\leq \sum_{j=1}^{2i-1}c(v_{j-1},v_j)\quad\text{and}\quad c(v_0,v_{2l+2-2i})\leq \sum_{j=1}^{2i-1}c(v_{2l+1-j},v_{(2l+2-j)\bmod (2l+1)}).
\end{equation}

Then, for any $i\in[t]$, we have
\begin{equation}\label{eq7}
\begin{split}
c(v_0, v_{2i-1})+c(v_0,v_{2l+2-2i})&\leq\sum_{j=1}^{2i-1}[c(v_{j-1},v_j)+c(v_{2l+1-j},v_{(2l+2-j)\bmod (2l+1)})]\\
&=\lrA{\sum_{j=1}^{i}\alpha_j+\sum_{j=1}^{i-1}\beta_j}c(T)\\
&=\lrA{\sum_{j=1}^{i}(\alpha_j+\beta_j)-\beta_i}c(T),
\end{split}    
\end{equation}
where the first inequality follows from (\ref{eq6}), and the first equality from the definitions of $\alpha_i$ and $\beta_i$.

Then, we have
\begin{equation}\label{eq8}
\begin{split}
&c(v_0,v_{2t-1})+c(v_0,v_{2t})+\sum_{i=1}^{t-1}\lrB{2c(v_0, v_{2i-1})+2c(v_0,v_{2l+2-2i})+\beta_ic(T)}\\
&\leq \lrA{\sum_{j=1}^{t}(\alpha_j+\beta_j)-\beta_t}c(T)+\sum_{i=1}^{t-1}\lrA{\sum_{j=1}^{i}2(\alpha_j+\beta_j)-\beta_i}c(T)\\
&=\lrA{\sum_{j=1}^{t}\alpha_j+\sum_{i=1}^{t-1}\sum_{j=1}^{i}2(\alpha_j+\beta_j)}c(T)\\
&=\lrA{\sum_{j=1}^{t}\alpha_j+\sum_{j=1}^{t}2(t-j)(\alpha_j+\beta_j)}c(T)\\
&=\lrA{\sum_{j=1}^{t}\alpha_j+2t-\sum_{j=1}^{t}2j(\alpha_j+\beta_j)}c(T),
\end{split}
\end{equation}
where the inequality follows from (\ref{eq7}) and the last equality from $\sum_{j=1}^{t}(\alpha_j+\beta_j)=1$.

Therefore, by (\ref{eq4}), (\ref{eq5}), and (\ref{eq8}), we obtain that
\begin{equation}\label{eq9}
\begin{split}
2\delta(E^*_T)&=c(E^*_T)+\sum_{i=1}^{2t-1}\lrA{c(v_0,v_i)+c(v_0,v_{2l+1-i})}\\
&\leq\lrA{\sum_{j=1}^{t}(\alpha_j+\beta_j)+2t-\sum_{j=1}^{t}2j(\alpha_j+\beta_j)}c(T)\\
&=\lrA{l+2-\sum_{j=1}^{t}2j(\alpha_j+\beta_j)}c(T)\\
&\leq\lrA{k+2-\sum_{j=1}^{t}2j\alpha_j-\sum_{j=1}^{t}2\beta_j}c(T)\\
&=\lrA{k+\sum_{j=1}^{t}2\alpha_j-\sum_{j=1}^{t}2j\alpha_j}c(T)\\
&=\lrA{k+2\alpha-\sum_{j=1}^{t}2j\alpha_j}c(T).
\end{split}
\end{equation}
where the first equality follows from (\ref{eq4}), the second equality follows from $\sum_{j=1}^{t}(\alpha_j+\beta_j)=1$ and $t=\frac{l+1}{2}$, the last equality follows from $\alpha=\sum_{j=1}^{t}\alpha_j$, the first inequality follows from (\ref{eq5}) and (\ref{eq8}) and $c(E^*_T)=\sum_{j=1}^{t}\beta_jc(T)$, and the second inequality follows from $l\leq k$ and $\beta_j\geq0$ for all $j\in[t]$.

Recall that $\alpha_j=0$ for $t<j\leq m$. Then, by (\ref{eq9}), we obtain $\delta(E^*_T)\leq\lrA{\frac{\capacity}{2}+\alpha-\sum_{i=1}^{m}i\alpha_i}c(T)$.
\end{proof}

By Claims~\ref{claim1} and \ref{claim2}, the properties (1) and (2) in Lemma~\ref{thebounds} are both satisfied when $\size{E^*_T}$ is odd. 

Therefore, Lemma~\ref{thebounds} holds when $\size{E^*_T}$ is odd. 

Next, we consider the case that $\size{E^*_T}$ is even.

\textbf{Case~2: $\size{E^*_T}$ is even.}
In this case, we have $t=\ceil{\frac{l+1}{2}}=\frac{l+2}{2}$.
We let
\[
\alpha_{i}=\frac{c(v_{2i-2},v_{2i-1})+c(v_{2l+2-2i},v_{(2l+3-2i)\bmod (2l+1)})}{c(T)},\quad\forall i\in[t-1],
\]
and
\[
\beta_{i}=\frac{c(v_{2i-1},v_{2i})+c(v_{2l+1-2i},v_{2l+2-2i})}{c(T)},\quad\forall i\in [t-1].
\]
Also, we let $\alpha_t=\frac{c(v_{2t-2},v_{2t-1})}{c(T)}$ and $\beta_t=0$. Then, we have $\alpha_i,\beta_i\geq0$ for all $i\in[t]$.

Similarly, $c(E(T)\setminus E^*_T)=\sum_{i=1}^t\alpha_ic(T)$ and $c(E^*_T)=\sum_{i=1}^t\beta_ic(T)$.
Then, $\sum_{i=1}^{t}(\alpha_i+\beta_i)=1$. 
Since $\alpha_i=0$ for any $t<i\leq m$, we have $\alpha=\sum^m_{i=1}\alpha_i=\sum_{i=1}^{t}\alpha_i=1-\sum_{i=1}^{t}\beta_i\leq 1$. 

Therefore, the properties (3) and (4) in Lemma~\ref{thebounds} are both satisfied.

Next, we prove the properties (1) and (2) in Lemma~\ref{thebounds}.

\begin{claim}\label{claim3}
When $\size{E^*_T}$ is even, it holds that $\MST_T\leq\lrA{1-\max_{\substack{1\leq i\leq m}}\frac{1}{2}\alpha_{i}}c(T)$.
\end{claim}
\begin{proof}[Claim Proof]
Recall that $\MST_T$ measures the cost of a minimum-cost spanning tree in $G[V_T]$ containing all edges in $E^*_T$.
Deleting an arbitrary edge in $E(T)\setminus E^*_T$ from $T$ forms a spanning tree in $G[V_T]$ containing all edges in $E^*_T$. Then, we have 
\begin{align*}
\MST_T&\leq c(T)-\max_{e\in E(T)\setminus E^*_T}c(e)\\
&= c(T)-\max_{1\leq i\leq 2t-1}c(v_{2i-2},v_{(2i-1)\bmod (2l+1)})\\
&\leq c(T)-\max_{\substack{1\leq i\leq t-1}}\max\lrC{\frac{1}{2}(c(v_{2i-2},v_{2i-1})+c(v_{2l+2-2i},v_{(2l+3-2i)\bmod (2l+1)})),c(v_{2t-2},v_{2t-1})}\\
&=\lrA{1-\max_{\substack{1\leq i\leq t-1}}\max\lrC{\frac{1}{2}\alpha_{i},\alpha_t}}c(T)\\
&\leq \lrA{1-\max_{\substack{1\leq i\leq m}}\frac{1}{2}\alpha_{i}}c(T),
\end{align*}
where the first equality follows from $E(T)\setminus E^*_T=\{(v_{2i-2},v_{(2i-1)\bmod (2l+1)})\mid i\in[2t-1]\}$, the second inequality from $E(T)\setminus E^*_T=\{(v_{2i-2},v_{2i-1}),(v_{2i-2},v_{(2i-1)\bmod (2l+1)})\mid i\in[t-1]\}\cup\{(v_{2t-2},v_{2t-1})\}$, and the last inequality from the definition of $\alpha_i$.
\end{proof}

\begin{claim}\label{claim4}
When $\size{E^*_T}$ is even, it holds that $\delta(E^*_T)\leq\lrA{\frac{\capacity}{2}+\alpha-\sum_{i=1}^{m}i\alpha_i}c(T)$.
\end{claim}
\begin{proof}[Claim Proof]
Since $E^*_T=\{(v_{2i-1},v_{2i})\mid i\in[2t-2]\}$, we have 
\begin{equation}\label{eeq4}
2\delta(E^*_T)=c(E^*_T)+\sum_{i=1}^{2t-2}\lrB{c(v_0, v_{2i-1})+c(v_0, v_{2i})}.%=c(E^*_T)+\sum_{i=1}^{2t-2}\lrA{c(v_0,v_i)+c(v_0,v_{2l+1-i})}.
\end{equation}

Since $E^*_T=\{(v_{2i-1},v_{2i}), c(v_{2l+1-2i},v_{2l+2-2i})\mid i\in[t-1]\}$, we have
\begin{equation}\label{eeq5}
\begin{split}
\sum_{i=1}^{2t-2}\lrB{c(v_0, v_{2i-1})+c(v_0, v_{2i})}&=\sum_{i=1}^{t-1}\lrB{c(v_0, v_{2i-1})+c(v_0, v_{2i})+c(v_0,v_{2l+1-2i})+c(v_0,v_{2l+2-2i})}\\
&\leq\sum_{i=1}^{t-1}\lrB{2c(v_0, v_{2i-1})+2c(v_0,v_{2l+2-2i})+\beta_ic(T)},
\end{split}
\end{equation}
where the inequality follows from $c(v_0, v_{2i})\leq c(v_0, v_{2i-1})+c(v_{2i-1},v_{2i})$ and $c(v_0,v_{2l+1-2i})\leq c(v_0,v_{2l+2-2i})+c(v_{2l+1-2i},v_{2l+2-2i})$ by the triangle inequality and $\beta_ic(T)=\frac{c(v_{2i-1},v_{2i})+c(v_{2l+1-2i},v_{2l+2-2i})}{c(T)}$ when $i\in[t-1]$.

For any $i\in[t-1]$, by the triangle inequality, we have
\begin{equation}\label{eeq6}
c(v_0, v_{2i-1})\leq \sum_{j=1}^{2i-1}c(v_{j-1},v_j)\quad\text{and}\quad c(v_0,v_{2l+2-2i})\leq \sum_{j=1}^{2i-1}c(v_{2l+1-j},v_{(2l+2-j)\bmod (2l+1)}).
\end{equation}

Then, for any $i\in[t-1]$, we have
\begin{equation}\label{eeq7}
\begin{split}
c(v_0, v_{2i-1})+c(v_0,v_{2l+2-2i})&\leq\sum_{j=1}^{2i-1}[c(v_{j-1},v_j)+c(v_{2l+1-j},v_{(2l+2-j)\bmod (2l+1)})]\\
&=\lrA{\sum_{j=1}^{i}\alpha_j+\sum_{j=1}^{i-1}\beta_j}c(T)\\
&=\lrA{\sum_{j=1}^{i}(\alpha_j+\beta_j)-\beta_i}c(T),
\end{split}    
\end{equation}
where the first inequality follows from (\ref{eeq6}), and the first equality from the definitions of $\alpha_i$ and $\beta_i$.

Then, we have
\begin{equation}\label{eeq8}
\begin{split}
&\sum_{i=1}^{t-1}\lrB{2c(v_0, v_{2i-1})+2c(v_0,v_{2l+2-2i})+\beta_ic(T)}\\
&\leq \sum_{i=1}^{t-1}\lrA{\sum_{j=1}^{i}2(\alpha_j+\beta_j)-\beta_i}c(T)\\
&=\lrA{\sum_{i=1}^{t-1}\sum_{j=1}^{i}2(\alpha_j+\beta_j)-\sum_{j=1}^{t-1}\beta_j}c(T)\\
&=\lrA{\sum_{j=1}^{t}2(t-j)(\alpha_j+\beta_j)-\sum_{j=1}^{t-1}\beta_j}c(T)\\
&=\lrA{2t-\sum_{j=1}^{t}2j(\alpha_j+\beta_j)-\sum_{j=1}^{t}\beta_j}c(T),
\end{split}
\end{equation}
where the inequality follows from (\ref{eeq7}) and the last equality from $\sum_{j=1}^{t}(\alpha_j+\beta_j)=1$ and $\beta_t=0$.

Therefore, by (\ref{eeq4}), (\ref{eeq5}), and (\ref{eeq8}), we obtain that
\begin{equation}\label{eeq9}
\begin{split}
2\delta(E^*_T)&=c(E^*_T)+\sum_{i=1}^{2t-2}\lrB{c(v_0,v_i)+c(v_0,v_{2l+1-i})}\\
&\leq\lrA{2t-\sum_{j=1}^{t}2j(\alpha_j+\beta_j)}c(T)\\
&=\lrA{l+2-\sum_{j=1}^{t}2j(\alpha_j+\beta_j)}c(T)\\
&\leq\lrA{k+2-\sum_{j=1}^{t}2j\alpha_j-\sum_{j=1}^{t}2\beta_j}c(T)\\
&=\lrA{k+\sum_{j=1}^{t}2\alpha_j-\sum_{j=1}^{t}2j\alpha_j}c(T)\\
&=\lrA{k+2\alpha-\sum_{j=1}^{t}2j\alpha_j}c(T).
\end{split}
\end{equation}
where the first equality follows from (\ref{eeq4}), the second equality follows from $\sum_{j=1}^{t}(\alpha_j+\beta_j)=1$ and $t=\frac{l+2}{2}$, the last equality follows from $\alpha=\sum_{j=1}^{t}\alpha_j$, the first inequality follows from (\ref{eeq5}) and (\ref{eeq8}) and $c(E^*_T)=\sum_{j=1}^{t}\beta_jc(T)$, and the second inequality follows from $l\leq k$ and $\beta_j\geq0$ for all $j\in[t]$.

Recall that $\alpha_j=0$ for $t<j\leq m$. Then, by (\ref{eeq9}), we have $\delta(E^*_T)\leq\lrA{\frac{\capacity}{2}+\alpha-\sum_{i=1}^{m}i\alpha_i}c(T)$.
\end{proof}

By Claims~\ref{claim3} and \ref{claim4}, the properties (1) and (2) in Lemma~\ref{thebounds} are both satisfied when $\size{E^*_T}$ is even. 

Therefore, Lemma~\ref{thebounds} also holds when $\size{E^*_T}$ is odd. 
\end{proof}

\section{A Note on the Approximation Ratio}\label{onlyone}
In this section, we show that the approximation ratio in Theorem~\ref{mainres} can be achieved by using only the RPP tour $H_1$ in $G$. 

By Lemma~\ref{rpp2}, it suffices to prove the following lemma.
\begin{lemma}\label{rpp3}
It holds that $c(H_1)\leq\OPT+2\MST-2c(E^*)$.
\end{lemma}
\begin{proof}
We first recall the details of the $\frac{3}{2}$-approximation algorithm~\citep{eiselt1995arc}.

The algorithm begins by computing a minimum-cost spanning tree $T^*$ in $G$ such that $E^*\subseteq E(T^*)$. Then, it computes a minimum-cost perfect matching $M^*$ in $G[Odd(V(T^*))]$, where $Odd(V(T^*))$ denotes the set of odd-degree vertices in $T^*$.
Clearly, the graph $G^*=(V\cup\{v_0\}, E(T^*)\cup E(M^*))$ forms an Eulerian graph.
Then, it computes an Eulerian tour in $G^*$, which then is transformed into the RPP tour $H_1$ by shortcutting. 
Therefore, we have 
\[
c(H_1)\leq c(E(T^*)) + c(M^*).
\]

By definition, we have $c(E(T^*))\leq \MST$.
Now, we prove that $c(M^*)\leq \OPT+\MST-2c(E^*)$.

Let $\overline{E}^*\coloneq E(T^*)\setminus E^*$ and $M$ be a minimum-cost perfect matching in $G[V]$. 

It is clear that $M\cup E^*$ forms a graph where each vertex has an even degree.
Moreover, since $\overline{E}^*\cup E^*=E(T^*)$, we know that $M\cup E^*\cup \overline{E}^*$ forms a connected graph. Furthermore, $M\cup E^*\cup \overline{E}^*\cup \overline{E}^*$ forms an Eulerian graph.
Therefore, the set $M\cup \overline{E}^*$ augment the tree $T^*$ into an Eulerian graph.

Since the cost function $c$ is a metric function, we known that $M^*$ augments $T^*$ into an Eulerian graph using minimum-cost. Therefore, we have 
\[
c(M^*)\leq c(M) + c(\overline{E}^*) = c(M) + c(E(T^*)) - c(E^*).
\]

Recall that $c(E(T^*))\leq \MST$ by definition and $c(E^*)+c(M)\leq\OPT$ by (\ref{eq2}). Thus, we obtain 
\[
c(M^*)\leq \OPT+\MST-2c(E^*),
\]
as desired.
\end{proof}

\section{Conclusion}
In this paper, by extending the techniques in approximating CVRP~\citep{DBLP:conf/mfcs/zhao}, we propose a $(\frac{5}{2}-\Theta(\frac{1}{\sqrt{\capacity}}))$-approximation algorithm for equal-demand CARP. To our knowledge, this is the first improvement over the classic result of $\frac{5}{2}-\frac{1.5}{\capacity}$~\citep{jansen1993bounds}.
In the future, it would be interesting to investigate whether the methods in~\citep{blauth2022improving} can be adapted to equal-demand CARP to break the $\frac{5}{2}$-approximation barrier.

\bibliographystyle{apalike}
\bibliography{main}

\begin{thebibliography}{}

\bibitem[Altinkemer and Gavish, 1987]{altinkemer1987heuristics}
Altinkemer, K. and Gavish, B. (1987).
\newblock Heuristics for unequal weight delivery problems with a fixed error guarantee.
\newblock {\em Oper. Res. Lett.}, 6(4):149--158.

\bibitem[Altinkemer and Gavish, 1990]{altinkemer1990heuristics}
Altinkemer, K. and Gavish, B. (1990).
\newblock Technical note—heuristics for delivery problems with constant error guarantees.
\newblock {\em Transportation Science}, 24(4):294--297.

\bibitem[Asano et~al., 1996]{AsanoKTT97}
Asano, T., Katoh, N., Tamaki, H., and Tokuyama, T. (1996).
\newblock Covering points in the plane by k-tours: a polynomial approximation scheme for fixed k.
\newblock {\em IBM Tokyo Research Laboratory Research Report RT0162}.

\bibitem[Babaee~Tirkolaee et~al., 2016]{babaee2016solving}
Babaee~Tirkolaee, E., Alinaghian, M., Bakhshi~Sasi, M., and Seyyed~Esfahani, M. (2016).
\newblock Solving a robust capacitated arc routing problem using a hybrid simulated annealing algorithm: a waste collection application.
\newblock {\em Journal of Industrial Engineering and Management Studies}, 3(1):61--76.

\bibitem[Belenguer and Benavent, 2003]{belenguer2003cutting}
Belenguer, J.~M. and Benavent, E. (2003).
\newblock A cutting plane algorithm for the capacitated arc routing problem.
\newblock {\em Computers \& Operations Research}, 30(5):705--728.

\bibitem[Blauth et~al., 2023]{blauth2022improving}
Blauth, J., Traub, V., and Vygen, J. (2023).
\newblock Improving the approximation ratio for capacitated vehicle routing.
\newblock {\em Math. Program.}, 197(2):451--497.

\bibitem[Bompadre et~al., 2006]{BompadreDO06}
Bompadre, A., Dror, M., and Orlin, J.~B. (2006).
\newblock Improved bounds for vehicle routing solutions.
\newblock {\em Discret. Optim.}, 3(4):299--316.

\bibitem[Brand{\~a}o and Eglese, 2008]{brandao2008deterministic}
Brand{\~a}o, J. and Eglese, R. (2008).
\newblock A deterministic tabu search algorithm for the capacitated arc routing problem.
\newblock {\em Computers \& Operations Research}, 35(4):1112--1126.

\bibitem[Chen et~al., 2016]{chen2016hybrid}
Chen, Y., Hao, J.-K., and Glover, F. (2016).
\newblock A hybrid metaheuristic approach for the capacitated arc routing problem.
\newblock {\em European Journal of Operational Research}, 253(1):25--39.

\bibitem[Christofides, 2022]{christofides1976worst}
Christofides, N. (2022).
\newblock Worst-case analysis of a new heuristic for the travelling salesman problem.
\newblock {\em Oper. Res. Forum}, 3(1).

\bibitem[Chv{\'{a}}tal, 1979]{chvatal1979greedy}
Chv{\'{a}}tal, V. (1979).
\newblock A greedy heuristic for the set-covering problem.
\newblock {\em Math. Oper. Res.}, 4(3):233--235.

\bibitem[Corber{\'a}n et~al., 2021]{corberan2021arc}
Corber{\'a}n, A., Eglese, R., Hasle, G., Plana, I., and Sanchis, J.~M. (2021).
\newblock Arc routing problems: A review of the past, present, and future.
\newblock {\em Networks}, 77(1):88--115.

\bibitem[Corber{\'a}n and Laporte, 2015]{corberan2015arc}
Corber{\'a}n, {\'A}. and Laporte, G. (2015).
\newblock {\em Arc routing: problems, methods, and applications}.
\newblock SIAM.

\bibitem[Corber{\'a}n and Prins, 2010]{corberan2010recent}
Corber{\'a}n, A. and Prins, C. (2010).
\newblock Recent results on arc routing problems: An annotated bibliography.
\newblock {\em Networks}, 56(1):50--69.

\bibitem[Dantzig and Ramser, 1959]{dantzig1959truck}
Dantzig, G.~B. and Ramser, J.~H. (1959).
\newblock The truck dispatching problem.
\newblock {\em Manag. Sci.}, 6(1):80--91.

\bibitem[Eiselt et~al., 1995a]{eiselt1995arc1}
Eiselt, H.~A., Gendreau, M., and Laporte, G. (1995a).
\newblock Arc routing problems, part {I}: The chinese postman problem.
\newblock {\em Oper. Res.}, 43(2):231--242.

\bibitem[Eiselt et~al., 1995b]{eiselt1995arc}
Eiselt, H.~A., Gendreau, M., and Laporte, G. (1995b).
\newblock Arc routing problems, part {II:} the rural postman problem.
\newblock {\em Oper. Res.}, 43(3):399--414.

\bibitem[Fern{\'a}ndez et~al., 2016]{fernandez2016collaboration}
Fern{\'a}ndez, E., Fontana, D., and Speranza, M.~G. (2016).
\newblock On the collaboration uncapacitated arc routing problem.
\newblock {\em Computers \& Operations Research}, 67:120--131.

\bibitem[Friggstad et~al., 2022]{uncvrp}
Friggstad, Z., Mousavi, R., Rahgoshay, M., and Salavatipour, M.~R. (2022).
\newblock Improved approximations for capacitated vehicle routing with unsplittable client demands.
\newblock In Aardal, K.~I. and Sanit{\`{a}}, L., editors, {\em Integer Programming and Combinatorial Optimization - 23rd International Conference, {IPCO} 2022, Eindhoven, The Netherlands, June 27-29, 2022, Proceedings}, volume 13265 of {\em Lecture Notes in Computer Science}, pages 251--261. Springer.

\bibitem[Golden and Wong, 1981]{GoldenW81}
Golden, B.~L. and Wong, R.~T. (1981).
\newblock Capacitated arc routing problems.
\newblock {\em Networks}, 11(3):305--315.

\bibitem[Gupta et~al., 2023]{gupta2023local}
Gupta, A., Lee, E., and Li, J. (2023).
\newblock A local search-based approach for set covering.
\newblock In Kavitha, T. and Mehlhorn, K., editors, {\em 2023 Symposium on Simplicity in Algorithms, {SOSA} 2023, Florence, Italy, January 23-25, 2023}, pages 1--11. SIAM, {SIAM}.

\bibitem[Haimovich and Kan, 1985]{HaimovichK85}
Haimovich, M. and Kan, A. H. G.~R. (1985).
\newblock Bounds and heuristics for capacitated routing problems.
\newblock {\em Math. Oper. Res.}, 10(4):527--542.

\bibitem[Harks et~al., 2013]{HarksKM13}
Harks, T., K{\"{o}}nig, F.~G., and Matuschke, J. (2013).
\newblock Approximation algorithms for capacitated location routing.
\newblock {\em Transportation Science}, 47(1):3--22.

\bibitem[Hassin and Levin, 2005]{hassin2005better}
Hassin, R. and Levin, A. (2005).
\newblock A better-than-greedy approximation algorithm for the minimum set cover problem.
\newblock {\em {SIAM} J. Comput.}, 35(1):189--200.

\bibitem[Jansen, 1993]{jansen1993bounds}
Jansen, K. (1993).
\newblock Bounds for the general capacitated routing problem.
\newblock {\em Networks}, 23(3):165--173.

\bibitem[Karlin et~al., 2021]{KarlinKG21}
Karlin, A.~R., Klein, N., and Gharan, S.~O. (2021).
\newblock A (slightly) improved approximation algorithm for metric {TSP}.
\newblock In Khuller, S. and Williams, V.~V., editors, {\em {STOC} '21: 53rd Annual {ACM} {SIGACT} Symposium on Theory of Computing, Virtual Event, Italy, June 21-25, 2021}, pages 32--45. {ACM}.

\bibitem[Karlin et~al., 2023]{DBLP:conf/ipco/KarlinKG23}
Karlin, A.~R., Klein, N., and Gharan, S.~O. (2023).
\newblock A deterministic better-than-3/2 approximation algorithm for metric {TSP}.
\newblock In Pia, A.~D. and Kaibel, V., editors, {\em Integer Programming and Combinatorial Optimization - 24th International Conference, {IPCO} 2023, Madison, WI, USA, June 21-23, 2023, Proceedings}, volume 13904 of {\em Lecture Notes in Computer Science}, pages 261--274. Springer.

\bibitem[Karpinski et~al., 2015]{karpinski2015new}
Karpinski, M., Lampis, M., and Schmied, R. (2015).
\newblock New inapproximability bounds for tsp.
\newblock {\em Journal of Computer and System Sciences}, 81(8):1665--1677.

\bibitem[Li and Simchi{-}Levi, 1990]{tight}
Li, C. and Simchi{-}Levi, D. (1990).
\newblock Worst-case analysis of heuristics for multidepot capacitated vehicle routing problems.
\newblock {\em {INFORMS} J. Comput.}, 2(1):64--73.

\bibitem[Mour{\~a}o et~al., 2009]{mourao2009heuristic}
Mour{\~a}o, M.~C., Nunes, A.~C., and Prins, C. (2009).
\newblock Heuristic methods for the sectoring arc routing problem.
\newblock {\em European Journal of Operational Research}, 196(3):856--868.

\bibitem[Mour{\~a}o and Pinto, 2017]{mourao2017updated}
Mour{\~a}o, M.~C. and Pinto, L.~S. (2017).
\newblock An updated annotated bibliography on arc routing problems.
\newblock {\em Networks}, 70(3):144--194.

\bibitem[Perrier et~al., 2007]{perrier2007survey}
Perrier, N., Langevin, A., and Campbell, J.~F. (2007).
\newblock A survey of models and algorithms for winter road maintenance. part iv: Vehicle routing and fleet sizing for plowing and snow disposal.
\newblock {\em Computers \& Operations Research}, 34(1):258--294.

\bibitem[Santos et~al., 2009]{santos2009improved}
Santos, L., Coutinho-Rodrigues, J., and Current, J.~R. (2009).
\newblock An improved heuristic for the capacitated arc routing problem.
\newblock {\em Computers \& Operations Research}, 36(9):2632--2637.

\bibitem[Schrijver, 2003]{schrijver2003combinatorial}
Schrijver, A. (2003).
\newblock {\em Combinatorial optimization: polyhedra and efficiency}, volume~24.
\newblock Springer.

\bibitem[Tang et~al., 2009]{tang2009memetic}
Tang, K., Mei, Y., and Yao, X. (2009).
\newblock Memetic algorithm with extended neighborhood search for capacitated arc routing problems.
\newblock {\em IEEE Transactions on Evolutionary Computation}, 13(5):1151--1166.

\bibitem[van Bevern et~al., 2020]{van2020approximate}
van Bevern, R., Fluschnik, T., and Tsidulko, O.~Y. (2020).
\newblock On approximate data reduction for the rural postman problem: Theory and experiments.
\newblock {\em Networks}, 76(4):485--508.

\bibitem[van Bevern et~al., 2014a]{van2014constant}
van Bevern, R., Hartung, S., Nichterlein, A., and Sorge, M. (2014a).
\newblock Constant-factor approximations for capacitated arc routing without triangle inequality.
\newblock {\em Oper. Res. Lett.}, 42(4):290--292.

\bibitem[Van~Bevern et~al., 2017]{van2017parameterized}
Van~Bevern, R., Komusiewicz, C., and Sorge, M. (2017).
\newblock A parameterized approximation algorithm for the mixed and windy capacitated arc routing problem: Theory and experiments.
\newblock {\em Networks}, 70(3):262--278.

\bibitem[van Bevern et~al., 2014b]{van2015chapter}
van Bevern, R., Niedermeier, R., Sorge, M., and Weller, M. (2014b).
\newblock The complexity of arc routing problems.
\newblock In {\em Arc routing: problems, methods, and applications}, pages 19--52. SIAM.

\bibitem[Williamson and Shmoys, 2011]{williamson2011design}
Williamson, D.~P. and Shmoys, D.~B. (2011).
\newblock {\em The design of approximation algorithms}.
\newblock Cambridge university press.

\bibitem[W{\o}hlk, 2008a]{wohlk2008approximation}
W{\o}hlk, S. (2008a).
\newblock An approximation algorithm for the capacitated arc routing problem.
\newblock {\em Open Operational Research Journal}, 2:8--12.

\bibitem[W{\o}hlk, 2008b]{wohlk2008decade}
W{\o}hlk, S. (2008b).
\newblock A decade of capacitated arc routing.
\newblock In {\em The vehicle routing problem: latest advances and new challenges}, pages 29--48. Springer.

\bibitem[W{\o}hlk and Laporte, 2018]{wohlk2018fast}
W{\o}hlk, S. and Laporte, G. (2018).
\newblock A fast heuristic for large-scale capacitated arc routing problems.
\newblock {\em Journal of the Operational Research Society}, 69(12):1877--1887.

\bibitem[Yu et~al., 2023]{yu2023exact}
Yu, W., Liao, Y., and Yang, Y. (2023).
\newblock Exact and approximation algorithms for the multi-depot capacitated arc routing problems.
\newblock {\em Tsinghua Science and Technology}, 28(5):916--928.

\bibitem[Zhao and Xiao, 2025a]{DBLP:conf/mfcs/zhao}
Zhao, J. and Xiao, M. (2025a).
\newblock Improved approximation algorithms for capacitated vehicle routing with fixed capacity.
\newblock In {\em 50th International Symposium on Mathematical Foundations of Computer Science, {MFCS} 2025, August 25-29, 2025, Warsaw, Poland}, volume 345 of {\em LIPIcs}, pages 93:1--93:19. Schloss Dagstuhl - Leibniz-Zentrum f{\"{u}}r Informatik.

\bibitem[Zhao and Xiao, 2025b]{zhao2025multidepot}
Zhao, J. and Xiao, M. (2025b).
\newblock Multidepot capacitated vehicle routing with improved approximation guarantees.
\newblock {\em Theor. Comput. Sci.}, 1043:115265.

\end{thebibliography}
\end{document}